\documentstyle[11pt]{article}

\begin{document}
\title{Many worlds and modality in the interpretation of\\quantum mechanics: an algebraic approach}
\author{{\sc G. Domenech} \thanks{%
Fellow of the Consejo Nacional de Investigaciones Cient\'{\i}ficas
y T\'ecnicas (CONICET) } $^{1,3}$, \  {\sc H. Freytes} \thanks{%
Permanent address, Instituto Argentino de Matem\'atica IAM-CONICET
 } $^{*,2}$ \ and \ {\sc C. de Ronde}$^{3,4}$}
\date{}
\maketitle

\begin{center}

\begin{small}
1. Instituto de Astronom\'{\i}a y F\'{\i}sica del Espacio (IAFE)\\
Casilla de Correo 67, Sucursal 28, 1428 Buenos Aires - Argentina\\
2. Universita degli Studi di Cagliari, Via Is Mirrionis 1, 09123,
Cagliari - Italia \\ 3. Center Leo Apostel (CLEA)\\ 4. Foundations
of  the Exact Sciences (FUND) \\ Brussels Free University
 Krijgskundestraat 33, 1160 Brussels - Belgium
\end{small}
\end{center}

\begin{abstract}

Many worlds interpretations (MWI) of quantum mechanics avoid the
measurement problem by considering every term in the quantum
superposition as actual. A seemingly opposed solution is proposed by
modal interpretations (MI) which state that quantum mechanics does
not provide an account of what `actually is the case', but rather
deals with what `might be the case', i.e. with possibilities. In
this paper we provide an algebraic framework which allows us to
analyze in depth the modal aspects of MWI.  Within our general
formal scheme we also provide a formal comparison between MWI and
MI, in particular, we provide a formal understanding of why ---even
though both interpretations share the same formal structure--- MI
fall pray of Kochen-Specker (KS) type contradictions while MWI
escape them.

\end{abstract}

\begin{small}
\centerline{\em Key words: Contextuality, Kochen Specker,
orthomodular lattices, modality, many worlds.}
\end{small}

\bibliography{pom}

\newtheorem{theo}{Theorem}[section]

\newtheorem{definition}[theo]{Definition}

\newtheorem{lem}[theo]{Lemma}

\newtheorem{prop}[theo]{Proposition}

\newtheorem{coro}[theo]{Corollary}

\newtheorem{exam}[theo]{Example}

\newtheorem{rema}[theo]{Remark}{\hspace*{4mm}}

\newtheorem{example}[theo]{Example}

\newcommand{\proof}{\noindent {\em Proof:\/}{\hspace*{4mm}}}

\newcommand{\qed}{\hfill$\Box$}

\newpage

\section{INTRODUCTION: MANY WORLDS AND \\ MODALITY}

Today, almost 50 years after its birth in 1957, the many worlds
interpretation (MWI) of quantum mechanics has become one of the most
important lines of investigation within the many interpretations of
quantum theory. MWI is considered to be a direct conclusion from
Everett's first proposal in terms of `relative states'
\cite{Everett57}. Everett's idea was to let quantum mechanics find
its own interpretation, making justice to the symmetries inherent in
the Hilbert space formalism in a simple and convincing way
\cite{DeWittGraham}. In this paper we will not address the main
argumentative lines of discussion raised for and against MWI (see
for example \cite{WHE, Deutsch85, Dieks07}). Rather, we shall
concentrate in its relation to the formal structure of quantum
mechanics and provide an algebraic frame which will allow us to
discuss the notion of {\it logical possibility} within it.

The main idea behind MWI is that superpositions refer to collections
of worlds, in each of which exactly one value of an observable,
which corresponds to one of the terms in the superposition, is
realized. Apart from being simple, the claim is that it possesses a
natural fit to the formalism, respecting its symmetries. This
provides a solution  to the measurement problem by assuming that
each one of the terms in the superposition is {\it actual} in its
own correspondent world. Thus, it is not only the single value which
we see in `our world' which gets actualized but rather, that a
branching of worlds takes place in every measurement, giving rise to
a multiplicity of worlds with their corresponding actual values. The
possible splits of the worlds are determined by the laws of quantum
mechanics.

Another proposed solution to the so called measurement problem has
been developed in the frame of modal interpretations (MI)
\cite{vF73, Kochen85, Dieks88}. According to these interpretations
``the quantum formalism does not tell us what actually is the case
in the physical world, but rather provides us with a list of
possibilities and their probabilities. The modal viewpoint is
therefore that quantum theory is about what may be the case, in
philosophical jargon, quantum theory is about modalities''
\cite{Dieks07}. Instead of actualizing every term in the
superposition, MI claim that each term {\it remains possible},
evolving with the Schr\"{o}dinger equation of motion.

Although MWI and MI share the same formal orthodox scheme, there are
but few comparisons in the literature \cite{Dieks07, Hemmothesis}.
In this paper we develop an algebraic framework which allows us to
analyze and discuss the modal aspects of MWI. Within this new formal
account, we can also provide a rigorous comparison between MWI and
MI. In particular, we can give a formal understanding of why MI fall
pray of KS-type contradictions \cite{KS, Bacc95} while MWI escape
them.

In Section 2, we introduce basic notions about lattice theory that
will be necessary later. In section 3, we provide a general
discussion on contextuality and modality in quantum mechanics. In
section 4, we develop a new algebraic frame for MWI. In section 5,
we formally compare MWI to MI.

\section{BASIC NOTIONS}

Now we recall from \cite{KAL} and \cite{MM}  some notions of lattice
theory that will play an important role in what follows. Let
$\mathcal{L}$ be a lattice and $a,b \in \mathcal{L}$. We say that
$b$ {\it covers} $a$ iff $a < b$ and moreover there exists no $x\in
\mathcal{L}$ such that $a<x<b$ for any $x$. Suppose that
$\mathcal{L}$ is a bounded lattice with $0$ the minimum element and
$1$ the maximum element. An element $p\in \mathcal{L}$ is called an
{\it atom} iff $p$ covers $0$ and a {\it coatom} iff $1$ covers $p$.
$\mathcal{L}$ is said to be an {\it atomistic lattice} iff for each
$x\in {\mathcal{L}}-\{0\}$, $x = \bigvee\{p\leq x: p \hspace{0.2cm}
is \hspace{0.2cm} an \hspace{0.2cm} atom \}$. An element $c\in
\mathcal{L}$ is said to be a {\it complement} of $a$ iff $a\land c =
0$ and $a\lor c = 1$. Let ${\mathcal{L}}=\langle
{\mathcal{L}},\lor,\land, 0, 1\rangle$ be a bounded lattice. Given
$a, b, c$ in $\mathcal{L}$, we write: $(a,b,c)D$\ \ iff $(a\lor
b)\land c = (a\land c)\lor (b\land c)$; $(a,b,c)D^{*}$ iff $(a\land
b)\lor c = (a\lor c)\land (b\lor c)$ and $(a,b,c)T$\ \ iff
$(a,b,c)D$, (a,b,c)$D^{*}$ hold for all permutations of $a, b, c$.
An element $z$ of a lattice $\mathcal{L}$ is called {\it central}
iff for all elements $a,b\in \mathcal{L}$ we have $(a,b,z)T$ and $z$
is complemented. We denote by $Z(\mathcal{L})$ the set of all
central elements of $\mathcal{L}$ and it is called the {\it center}
of $\mathcal{L}$.

A {\it  lattice with involution} \cite{Ka} is an algebra $\langle
\mathcal{L}, \lor, \land, \neg \rangle$ such that $\langle
\mathcal{L}, \lor, \land \rangle$ is a  lattice and $\neg$ is a
unary operation on $\mathcal{L}$ that fulfills the following
conditions: $\neg \neg x = x$ and $\neg (x \lor y) = \neg x \land
\neg y$. An {\it orthomodular lattice} is an algebra $\langle
{\mathcal{L}}, \land, \lor, \neg, 0,1 \rangle$ of type $\langle
2,2,1,0,0 \rangle$ that satisfies the following conditions

\begin{enumerate}
\item
$\langle {\mathcal{L}}, \land, \lor, \neg, 0,1 \rangle$ is a bounded
lattice with involution,

\item
$x\land  \neg x = 0 $.

\item
$x\lor ( \neg x \land (x\lor y)) = x\lor y $

\end{enumerate}

We denote by ${\cal OML}$ the variety of orthomodular lattices. It
is well known that if $\mathcal{H}$ is a Hilbert space then
$\mathcal{L}(\mathcal{H})$, the lattice of closed subspaces of
$\mathcal{H}$, is an atomistic orthomodular lattice. {\it Boolean
algebras} are orthomodular lattices satisfying the {\it distributive
law} $x\land (y \lor z) = (x \land y) \lor (x \land z)$. We denote
by ${\bf 2}$ the Boolean algebra of two elements. If $\mathcal{L}$
is a bounded lattice then $Z(\mathcal{L})$ is a Boolean sublattice
of $\mathcal{L}$ {\rm \cite[Theorem 4.15]{MM}}.

Let $A$ be a Boolean algebra. A subset $F$ of $A$ is called a {\it
filter} iff it satisfies: if $a\in F$ and $a\leq x$ then $x\in F$
and if $a,b\in F$ then $a\land b \in F$. $F$ is a {\it proper
filter} iff $F\not = A$ or, equivalently, $0\not \in F$. If
$X\subseteq A$, the filter $F_X$ generated by $X$ is the minimum
filter containing $X$. It is well know that $F_X = \{x\in A: \exists
x_1 \cdots x_n \in X \hspace{0.2cm} with \hspace{0.2cm} x_1 \land
\cdots \land x_n \leq x \}$. Each filter $F$ in $A$ determines
univocally a congruence in which the equivalence classes are given
by $[x] = \{y \in A: \neg x \lor y \in F \hspace{0.2cm} and
\hspace{0.2cm} x \lor \neg y \in F \}$. In this case the quotient
set $A/_{\sim}$, noted as $A/F$, is a Boolean algebra and the
natural application $x \mapsto [x]$ is a Boolean homomorphism form
$A$ to $A/F$.  A proper filter $F$ is {\it maximal} iff the quotient
algebra $A/F$ is isomorphic to $\bf 2$. It is well known that each
proper filter can be extended to a maximal one. A very important
property associated with maximal filters is the following: suppose
that $x \not \leq y$. Then there exists a maximal filter $F$ in $A$
such that $x\in F$ and $y \not \in F$. We will refer to this result
as the {\it maximal filter theorem}.

\section{CONTEXTUALITY AND MODALITY IN \\ QUANTUM SYSTEMS}

In the usual terms of quantum logic \cite{ByvN, Jauch68}, a property
of a system is related to a subspace of the Hilbert space ${\mathcal
H}$ of its (pure) states or, analogously, to the projector operator
onto that subspace. A physical magnitude ${\mathcal M}$ is
represented by an operator $\bf M$ acting over the state space. For
bounded self-adjoint operators, conditions for the existence of the
spectral decomposition ${\bf M}=\sum_{i} a_i {\bf P}_i=\sum_{i} a_i
|a_i><a_i|$ are satisfied. The real numbers $a_i$ are related to the
outcomes of measurements of the magnitude ${\mathcal M}$ and
projectors $|a_i><a_i|$ to the mentioned properties.  Thus, the
physical properties of the system are organized in the lattice of
closed subspaces ${\mathcal L}({\mathcal H})$. Moreover, each
self-adjoint operator $\bf M$  has associated a Boolean sublattice
$W_{\bf{M}}$ of $L({\mathcal H})$ which we will refer to as the
spectral algebra of the operator $\bf M$.

Assigning values to a physical quantity ${\cal M}$ is equivalent to
establishing a Boolean homomorphism $v: W_{\bf{M}} \rightarrow {\bf
2}$. Thus, we can say that it makes sense to use the ``classical
discourse'' ---this is, the classical logical laws are valid---
within the context given by ${\mathcal M}$.

One may define a {\it global valuation} of the physical magnitudes
over ${\mathcal L}({\mathcal H})$ as a family of Boolean
homomorphisms $(v_i: W_i \rightarrow {\bf 2})_{i\in I}$ such that
$v_i\mid W_i \cap W_j = v_j\mid W_i \cap W_j$ for each $i,j \in I$,
being $(W_i)_{i\in I}$ the family of Boolean sublattices of
${\mathcal L}({\mathcal H})$. This global valuation would give the
values of all magnitudes at the same time maintaining a {\it
compatibility condition} in the sense that whenever two magnitudes
shear one or more projectors, the values assigned to those
projectors are the same from every context. As we have proved in
\cite{DF}, the KS theorem in the algebraic terms of the previous
definition rules out this possibility:

\begin{theo}\label{CS3}
If $\mathcal{H}$ is a Hilbert space such that $dim({\mathcal{H}}) >
2$, then a global valuation over ${\mathcal L}({\mathcal H})$ is not
possible.\qed
\end{theo}

This impossibility to assign values to the properties at the same
time satisfying compatibility conditions is a weighty obstacle for
the interpretation of the formalism.\\

We have introduced elsewhere \cite{DFR1, DFR2} a general modal
scheme which extends the expressive power of the orthomodular
structure to provide a rigorous framework for the Born rule and
mainly, to discuss the restrictions posed by the KS theorem to {\it
possible properties}. We recall here some notions that will be
useful in our development.

First, we enriched the orthomodular structure with a modal operator
taking into account the following considerations:

\begin{enumerate}

\item
Propositions about the properties of the physical system are
interpreted in the orthomodular lattice of closed subspaces of
$\mathcal{H}$. Thus, we retain this structure in our extension.

\item
Given a proposition about the system, it is possible to define a
context from which one can predicate with certainty about it
together with a set of propositions that are compatible with it and,
at the same time, predicate probabilities about the other ones (Born
rule). In other words, one may predicate truth or falsity of all
possibilities at the same time, i.e. possibilities allow an
interpretation in a Boolean algebra. In rigorous terms, for each
proposition $P$, if we refer with $\Diamond P$ to the possibility of
$P$, then $\Diamond P$ will be a central element of a orthomodular
structure.

\item
If $P$ is a proposition about the system and $P$ occurs, then it
is trivially possible that $P$ occurs. This is expressed as $P \leq
\Diamond P$. \hspace{0.2cm}

\item
Assuming an actual property and a complete set of properties that
are compatible with it determines a context in which the classical
discourse holds. Classical consequences that are compatible with it,
for example probability assignments to the actuality of other
propositions, shear the classical frame. These consequences are the
same ones as those which would be obtained by considering the
original actual property as a possible one. This is interpreted in
the following way: if $P$ is a property of the system, $\Diamond P$
is the smallest central element greater than $P$.

\end{enumerate}

From consideration 1, it follows that the original orthomodular
structure is maintained. The other considerations are satisfied if
we consider a modal operator $\Diamond$ over an orthomodular lattice
$\mathcal{L}$  defined as $$\Diamond a = Min \{z\in
Z({\mathcal{L}}): a\leq z \}$$ with $Z({\mathcal{L}})$ the center of
$\mathcal{L}$. When this minimum exists for each $a\in \mathcal{L}$
we say that $\mathcal{L}$ is a {\it Boolean saturated orthomodular
lattice}. On each Boolean saturated orthomodular lattice  we can
define the necessity operator as a unary operation $\Box$ given by
$\Box x = \neg \Diamond \neg x$. We have shown that this enriched
orthomodular structure can be axiomatized by equations conforming a
variety denoted by ${\cal OML}^\Diamond$ \cite{DFR1}. More
precisely, each element of ${\cal OML}^\Diamond$ is an algebra $
\langle {\mathcal{L}}, \land, \lor, \neg, \Box, 0, 1 \rangle$ of
type $ \langle 2, 2, 1, 1, 0, 0 \rangle$ satisfying the following
equations:

\begin{enumerate}

\item[S1]
$\Box x \leq x$ \hspace{3cm} S5 \hspace{0.2cm} $y = (y\land \Box x) \lor (y \land \neg \Box x)$

\item[S2]
$\Box 1 = 1$  \hspace{3.1 cm} S6 \hspace{0.2cm} $\Box (x \lor \Box y ) = \Box x \lor \Box y $

\item[S3]
$\Box \Box x = \Box x$  \hspace{2.5cm} S7 \hspace{0.2cm} $\Box (\neg x \lor (y \land x)) \leq \neg \Box x \lor \Box y $

\item[S4]
$\Box(x \land y) = \Box(x) \land \Box(y)$

\end{enumerate}

Orthomodular complete lattices are examples of Boolean saturated
orthomodular lattices and we can embed each orthomodular lattice
$\mathcal{L}$ in an element $\mathcal{L}^{\Diamond} \in  {\cal
OML}^\Diamond$. In general, $\mathcal{L}^{\Diamond}$ is referred as
a {\it modal extension of $\mathcal{L}$}. In this case we may see
the lattice $\mathcal{L}$ as a subset of $\mathcal{L}^{\Diamond}$
(see \cite{DFR1}).

\begin{definition}
{\rm Let $\mathcal{L}$ be an orthomodular lattice and
$\mathcal{L}^\Diamond \in {\cal OML}^\Diamond$ be a modal extension
of $\mathcal{L}$. We define the {\it possibility space} of
$\mathcal{L}$ in ${\mathcal{L}}^\Diamond$ as $$\Diamond
{\mathcal{L}} = \langle \{\Diamond p : p \in \mathcal{L} \}
\rangle_{\mathcal{L}^\Diamond}$$ }
\end{definition}

\noindent The {\it possibility space} represents the modal content
added to the discourse about properties of the system.

\begin{prop}\label{POSSPACE}{\rm \cite[Proposition 14]{DFR1}}
Let $\mathcal{L}$ be an orthomodular lattice, $W$  a Boolean
sublattice of $\mathcal{L}$ and $\mathcal{L}^\Diamond \in {\cal
OML}^\Diamond$ a modal extension of $\mathcal{L}$. Then $\langle W
\cup \Diamond \mathcal{L} \rangle_{\mathcal{L}^\Diamond}$ is a
Boolean sublattice of $\mathcal{L}^\Diamond$. In particular
$\Diamond \mathcal{L}$ is a Boolean sublattice of
$Z(\mathcal{L}^\Diamond)$. \qed
\end{prop}

Now, we develop the algebraic counterpart of the classical notion of
{\it consequence} which will be useful when formalizing the concept
of possibility in MWI. As will become clear below, Proposition
\ref{POSSPACE} allows to establish a deep relation between this
concept and the possibility space.

\begin{definition}
{\rm Let $\mathcal{L}$ be an orthomodular lattice, $p\in
\mathcal{L}$ and $\mathcal{L}^\Diamond \in {\cal OML}^\Diamond$ a
modal extension of $\mathcal{L}$. Then $x \in \Diamond \mathcal{L}$
is said to be a {\it classical consequence} of $p$ iff for each
Boolean sublattice $W$ in $\mathcal{L}$ (with $p\in W$) and each
Boolean valuation $v:W \rightarrow {\bf 2}$, $v(x) = 1$ whenever
$v(p)=1$. We denote by $Cons_{\mathcal{L}^\Diamond}(p)$ the set of
classical consequences of $\mathcal{L}$. }
\end{definition}

\begin{prop}\label{CLASCONS}
Let $\mathcal{L}$ be an orthomodular lattice, $p\in \mathcal{L}$ and
$\mathcal{L}^\Diamond \in {\cal OML}^\Diamond$ a modal extension of
$\mathcal{L}$. Then we have that
$$Cons_{{\mathcal{L}}^\Diamond}(p) = \{x\in \Diamond {\mathcal{L}}: p \leq x \} = \{x\in
\Diamond {\mathcal{L}}: \Diamond p \leq x \} $$
\end{prop}

\begin{proof}
By definition of $\Diamond$ it is clear that  $\{x\in \Diamond
{\mathcal{L}}: p \leq x \} = \{x\in \Diamond {\mathcal{L}}: \Diamond
p \leq x \}$ and the inclusion $\{x\in \Diamond {\mathcal{L}}:
\Diamond p \leq x \} \subseteq Cons_{{\mathcal{L}}^\Diamond}(p)$ is
trivial. Let $x\in Cons_{\mathcal{L}^\Diamond}(p)$. Assume that $p
\not \leq x$. Consider the Boolean sub algebra of $\mathcal{L}$
given by $W = \{p, \neg p, 0,1 \}$. By Proposition \ref{POSSPACE},
$W^\Diamond = \langle W \cup \Diamond \mathcal{L}
\rangle_{\mathcal{L}^\Diamond}$ is a Boolean sublattice of
$\mathcal{L}^\Diamond$. By the maximal filter theorem, there exists
a maximal filter $F$ in $W^\Diamond$ such that $p\in F$ and $x \not
\in F$. If we consider the quotient Boolean algebra $W^\Diamond /F$
and the natural Boolean homomorphism $f: W^\Diamond \rightarrow
W^\Diamond/F = {\bf 2}$, then $f(p) = 1$ and $f(x) = 0$, which is a
contradiction. \qed
\end{proof}

Let $\mathcal{L}$ be an orthomodular lattice, $(W_i)_{i \in I}$ the
family of Boolean sublattices of $\mathcal{L}$ and
$\mathcal{L}^\Diamond$ a modal extension of $\mathcal{L}$. If $f:
\Diamond \mathcal{L} \rightarrow {\bf 2}$ is a Boolean homomorphism,
an {\it actualization compatible} with  $f$ is a global valuation
$(v_i: W_i \rightarrow {\bf 2})_{i\in I}$ such that $v_i \mid W_i \cap \Diamond {\cal L}  = f \mid  W_i \cap \Diamond {\cal L} $ for each $i\in I$. Compatible actualizations represent the passage
from possibility to actuality.

\begin{theo}\label{MKS}{\rm \cite[Theorem 19]{DFR1}}
Let $\mathcal{L}$ be an orthomodular lattice. Then $\mathcal{L}$
admits a global valuation iff for each possibility space there
exists a Boolean homomorphism  $f: \Diamond \mathcal{L} \rightarrow
{\bf 2}$ that admits  a compatible actualization. \qed
\end{theo}

The addition of modalities to the discourse about the properties
of a quantum system enlarges its expressive power. At first sight
it may be thought that this could help to circumvent
contextuality, allowing to refer to physical properties belonging
to the system in an objective way that resembles the classical
picture. Since the possibility space is a Boolean algebra, there
exists a Boolean valuation of the possible properties. But in view
of the last theorem, a global actualization that would correspond
to a family of compatible  valuations is prohibited. Thus, the
theorem states that {\it the contextual character is maintained
even when the discourse is enriched with modalities.}

\section{AN ALGEBRAIC FRAME FOR MANY \\ WORLDS}

In the MWI, all possibilities encoded in the wave function take
place, but in different worlds. More precisely: let ${\cal M}$ be a
physical magnitude represented by an operator ${\bf M}$ with
spectral decomposition ${\bf M}=\sum_{i} a_i {\bf P}_i$. If a
measurement of ${\bf M}$ is performed and $a_1$ occurs, then in
another world $a_2$ occurs, and in some other world $a_3$ occurs,
etc. Let us now see  how we can introduce our modal algebraic frame
for MWI.

Let ${\mathcal{H}}$ be a Hilbert space and suppose that ${\bf M}$
has associated a Boolean sublattice $W_{\bf M}$ of
$\mathcal{L}({\mathcal{H}})$. The family $({\bf P}_i)_i$ is
identified as elements of $W_{\bf M}$. If a measurement is performed
and its result is $a_i$, this means that we can establish a Boolean
homomorphism $$v_i:W_{\bf M} \rightarrow {\bf 2} \hspace{0.7cm}s.t.
\hspace{0.2cm} v_i({\bf P}_i) = 1 $$

\subsection{${\cal OML}^\Diamond$-CONSEQUENCES}
In a possible world where $v_i({\bf P}_i) = 1$ we will have
classical consequences. We can take an arbitrary modal extension
$\mathcal{L}^{\Diamond}$ of $\mathcal{L}({ \mathcal{H}})$ and
consider the set $Cons_{\mathcal{L}^\Diamond}({\bf P}_i)$. The modal
extension {\it does not depend} on the valuation over the family
$({\bf P}_i)_i$. Thus, it is clear that the modal extension is
independent of any possible world. Modal extensions are simple
algebraic extensions of an orthomodular structure. By Proposition
\ref{CLASCONS} we have that $Cons_{{\mathcal{L}}^\Diamond}({\bf
P}_i) = \{x\in \Diamond {\mathcal{L}}({\mathcal{H}}): \Diamond {\bf
P}_i \leq x \} $. Thus, for any arbitrary modal extension
$\mathcal{L}^{\Diamond}$ of $\mathcal{L}({\mathcal{H}})$ in terms of
classical consequences, the classical consequences of $v_i({\bf
P}_i) = 1$ are exactly the same ones as $\Diamond {\bf P}_i$
(independently of any possible splitting). In terms of classical
consequences which refer to a property ${\bf P}_i$, it is the same
to consider the classical consequences in the possible world where
$v_i({\bf P}_i) = 1$, than to study the classical consequences of
$\Diamond {\bf P}_i$ before the splitting.

MWI {\it maintains that in each respective i-world, $v_i({\bf P}_i)
= 1$ for each $i$}. Thus, a family of valuations $(v_i({\bf P}_i) =
1)_i$ may be simultaneously considered, each member being realized
in each different $i$-world. From an algebraic perspective, this
would be equivalent to have a family of pairs $\langle
{\mathcal{L}}({\mathcal{H}}), v_i({\bf P}_i) = 1 \rangle_i$, each
pair being the orthomodular structure ${\mathcal{L}}({\mathcal{H}})$
with a distinguished Boolean valuation $v_i$ over a spectral
sub-algebra containing ${\bf P}_i$ such that $v_i({\bf P}_i) = 1$.
In what follows, we will show that the  ${\cal OML}^\Diamond$
structure is able to capture  this fact in terms of classical
consequences. For this purpose, the following proposition is needed.

\begin{prop}\label{COMPL2}
Let ${\mathcal{H}}$ be Hilbert space such that $dim({\mathcal{H}})
> 2$ and $a,b$ be a two distinct atoms in $\mathcal{L}({\mathcal{H}})$.
If we consider a modal extension $\mathcal{L}^{\diamond}$ of
$\mathcal{L}({\mathcal{H}})$, then $\Diamond(a) = \Diamond(b)$.
\end{prop}

\begin{proof}
We first note that there exists a coatom $c$ such that $c$ is not
comparable with $a$ and $b$. In fact, let $(c_i)_{i\in I}$ be the
family of coatoms of $\mathcal{L}({\mathcal{H}})$ and suppose that
$a\leq c_i$ and $b \leq c_i$ for each ${i\in I}$. We then get $a\lor
b \leq \bigcap_{i\in I} c_i = 0$, which is a contradiction. Since
$a$ and $b$ are atoms and $c$ is a coatom not comparable with $a$
and $b$ then $0 = a\land c = b\land c$ and $1 = a\lor c = b\lor c$.
Hence $c$ is a common complement of $a,b$. Since
$\mathcal{L}({\mathcal{H}})\hookrightarrow \mathcal{L}^{\Diamond}$
is an ${\cal OML}$-embeding, $c$ is a common complement of $a,b$ in
$\mathcal{L}^{\Diamond}$. We first note $\neg \Diamond \neg c = \Box
c \leq c$, then $a\land \neg \Diamond \neg c \leq a\land  c = 0$.
Since $\neg \Diamond \neg c$ is a central element, $a \leq \Diamond
\neg c$ and $\Diamond a \leq \Diamond \neg c$. Since $a\lor c = 1$
then $\neg a \land \neg c = 0$. Therefore $\neg c \land \neg
\Diamond a \leq \neg a \land \neg c = 0$. Since $\neg \Diamond a$ is
central element then $\neg c \leq \Diamond a$ and $\Diamond \neg c
\leq \Diamond a$. With the same argument we can prove that
$\Diamond(b) = \Diamond(\neg c)$.

\qed
\end{proof}

The following theorem is crucial in order  to relate MWI with
modality in terms of valuations and classical consequences.

\begin{theo}\label{MANY}
Let ${\mathcal{H}}$ be Hilbert space such that $dim({\mathcal{H}})
> 2$ and $(p_i)_{i \in I}$ be a family of elements of
$\mathcal{L}({\mathcal{H}})$. If we consider a modal extension
$\mathcal{L}({\mathcal{H}})\hookrightarrow \mathcal{L}^{\Diamond}$
then there exists a Boolean homomorphism $v: \Diamond \mathcal{L}
\rightarrow {\bf 2}$ such that $v(\Diamond p_i) = 1$ for each $i\in
I$.
\end{theo}

\begin{proof}
Since $\mathcal{L}({\mathcal{H}})$ is an atomic lattice, for each
$p_i$ there exists an atom $a_i$ such that $a_i \leq p_i$. Let $I_0$
be a finite subfamily of $I$. By Proposition \ref{COMPL2}, we have
that $0 < \bigwedge_{i\in I_0} \Diamond(a_i) \leq \bigwedge_{i\in
I_0} \Diamond(p_i)$. Therefore the family $(\Diamond p_i)_{\i \in
I}$ generates a proper filter $F$ in the Boolean algebra $\Diamond
\mathcal{L}$. Extending $F$ to a maximal filter $F_{\bf M}$, the
natural Boolean homomorphism $v: \Diamond \mathcal{L} \rightarrow
{\bf 2}$ satisfies that for each $i\in I$, $v(\Diamond p_i) = 1$.

\qed
\end{proof}

While MWI considers a family of pairs $\langle
{\mathcal{L}}({\mathcal{H}}), v_i({\bf P}_i) = 1 \rangle_i$ for each
possible i-world and the classical consequences of $v_i({\bf P}_i) =
1$ in the $i$-world, the ${\cal OML}^\Diamond$ structure, by
Proposition \ref{CLASCONS}, considers classical consequences of each
$v_i({\bf P}_i) = 1$ coexisting simultaneously in one and the same
structure, what is possible in view of Theorem \ref{MANY}. More
precisely, as a valuation $v: \Diamond \mathcal{L} \rightarrow {\bf
2}$ exists such that $v(\Diamond {\bf P}_i) = 1$ for each $i$, each
element $x\in \Diamond \mathcal{L}$ such that ${\bf P}_i \leq x$
necessarily satisfies $v(x)=1$.

\subsection{MANY WORLDS AND KOCHEN-SPECKER TYPE THEOREMS}

KS theorem does not impose conditions on both the  family of
valuations $v_i({\bf P}_i) = 1$, considered as a family of pairs
$\langle {\mathcal{L}}({\mathcal{H}}), v_i({\bf P}_i) = 1 \rangle_i$
in MWI nor on the  Boolean valuation $v: \Diamond
\mathcal{L}({\mathcal{H}}) \rightarrow {\bf 2}$ satisfying
$v(\Diamond {\bf P}_i) = 1$ for each $i$ in the ${\cal
OML}^\Diamond$ structure (Theorem \ref{MANY}). In fact, by Theorem
\ref{MKS} KS only prevents from extending the valuation $v: \Diamond
\mathcal{L}({\mathcal{H}}) \rightarrow {\bf 2}$ to
$\mathcal{L}({\mathcal{H}})$ in a compatible manner. In the wording
previous to Theorem \ref{MKS}, KS theorem prohibits to pass from the
realm of possibility to that of actuality in the sense that it
precludes to establish a compatible actualization for $v: \Diamond
\mathcal{L}({\mathcal{H}}) \rightarrow {\bf 2}$ to
$\mathcal{L}({\mathcal{H}})$ when $dim({\mathcal{H}})> 2$ (see
Theorem \ref{CS3}).

In its algebraic version given in Theorem \ref{CS3}, KS {\it only
imposes a limit to the possibility of establishing compatible
valuations over $\mathcal{L}({\mathcal{H}})$ when
$dim({\mathcal{H}})> 2$} but does not cause incompatibilities when
reference is made to possible global valuations in the realm of
possibilities considered in the ${\cal OML}^\Diamond$ structure or
in the family $\langle {\mathcal{L}}({\mathcal{H}}), v_i({\bf P}_i)
= 1 \rangle_i$ from MWI. Thus, valuations over different i-worlds
are admitted.

\section{CONCLUSIONS}

In this paper we have analyzed the orthomodular formal structure of
quantum mechanics  in relation to both  MWI and MI. In order to deal
with logical possibility in these interpretations, we considered two
different algebraic approaches which were characterized in Sec. 3
and 4. In the case of MWI, the structure is the family of pairs
$<{\mathcal{L}}({\mathcal{H}}), v_{i}(P_{i})=1>_{i}$ of orthomodular
lattices with a distinguished Boolean valuation that assigns
``true'' to a projector of a spectral algebra in each one of them.
For MI, we have the Boolean saturated orthomodular lattice. Both
structures allows us to compare the role of contextuality in
relation to the formal account of actual and possible properties in
a rigorous way  as it is shown in Table 1.

The modal scheme we developed in  \cite{DFR1},  i.e. the Boolean
saturated orthomodular lattice, is also adequate to consider the
notion of possibility within MWI. The whole set of possible worlds,
each one with an actualized value of a property, is algebraically
equivalent to the set of valuations to ``true'' of the possible
properties in $\mathcal{L}^{\Diamond}$. That is to say, the
actualization of each value of a given property in each i-world is
analogous to the assignment of the value ``true'' to all possible
properties in the scheme of MI.

We have shown that the KS theorem only imposes a limit to the
possibility of establishing compatible valuations over
$\mathcal{L}({\mathcal{H}})$. However, there is no incompatibility
when reference is only made to valuations in the realm of
possibilities, i.e., in the ${\cal OML}^\Diamond$ structure for MI
or in the family $\langle {\mathcal{L}}({\mathcal{H}}), v_i({\bf
P}_i) = 1 \rangle_i$ for MWI.

\begin{table}\label{TABLA}
\begin{center}
\begin{tabular}{cp{2in}p{2in}}
& \multicolumn{1}{c}{\it Modality}
& \multicolumn{1}{c}{\it MWI} \\
\it Valuations & There exists a Boolean valuation
$v:\Diamond(\mathcal{L}({\mathcal{H}})) \rightarrow {\bf 2}$ such
that $v(\Diamond P_i) = 1$ for each $i$. &
Families of Boolean valuations $(v_i(P_i) = 1)_i$, may be simultaneously considered, each member being realized in each different $i$-world. \vspace{0.3cm}\\

\it KS theorem & precludes to establish compatible actualizations
for $v:\Diamond(\mathcal{L}({\mathcal{H}})) \rightarrow {\bf 2}$ to
$\mathcal{L}({\mathcal{H}})$. & does not cause incompatibility when
each member of a family of valuations $(v_i(P_i) = 1)_i$ is
considered.
\end{tabular}
\end{center}
\caption {Modality and MWI under ${\cal OML}^{\diamond}$-structure}
\label{compar}
\end{table}

\section*{Acknowledgements}

This work was partially supported by the following grants: PICT
04-17687 (ANPCyT), PIP N$^o$ 6461/05 (CONICET), UBACyT N$^o$ X081
and X204 and Projects of the Fund for Scientific Research Flanders
G.0362.03 and G.0452.04.

\end{document}